\documentclass[sn-mathphys-num]{sn-jnl}
\usepackage{graphicx}
\usepackage{multirow}
\usepackage{amsmath,amssymb,amsfonts}
\usepackage{amsthm}
\usepackage{mathrsfs}
\usepackage[title]{appendix}
\usepackage{xcolor}
\usepackage{textcomp}
\usepackage{manyfoot}
\usepackage{booktabs}
\usepackage{algorithm}
\usepackage{algorithmicx}
\usepackage{quiver}
\usepackage{algpseudocode}
\usepackage{listings}
\usepackage{physics}

\def\nn{\nonumber}

\def\f{\frac}
\def\l{\left}
\def\r{\right}

\def\R{\mathbb{R}}

\def\Z{\mathbb{Z}}

\def\p{\phi}
\def\a{\alpha}

\def\t{\theta}

\def\ie{\textit{i.e.}}
\theoremstyle{thmstyleone}
\newtheorem{theorem}{Theorem}
\theoremstyle{thmstyletwo}%

\theoremstyle{thmstylethree}%
\newtheorem{definition}{Definition}%

\begin{document}

\title[Article Title]{Group theoretic quantization of punctured plane}
\author*[1]{\fnm{Manvendra} \sur{Somvanshi}}\email{ms20126@iisermohali.ac.in}
\author[2]{\fnm{D. Jaffino} \sur{Stargen}}\email{jaffinostargend@gmail.com}
%

\affil*[1]{\orgdiv{Department of Mathematical Sciences},
\orgname{Indian Institute of Science Education \& Research Mohali}, \orgaddress{\street{Sector 81},
\city{SAS Nagar}, \postcode{140306}, \state{Punjab}, \country{India}}}

%

\affil[2]{Department of Physics, Global Academy of Technology,
Aditya Layout, RR Nagar, Bengaluru, Karnataka 560098, India}

\affil[2]{\orgdiv{Department of Mechanical Engineering}, \orgname{Massachusetts Institute of Technology},
\orgaddress{\street{77 Massachusetts Avenue}, \city{Cambridge}, \postcode{02139}, \state{MA}, \country{USA}}}

\abstract{
We quantize punctured plane, $X=\mathbb{R}^2-\{0\}$, employing Isham's group theoretic quantization
procedure. After sketching out a brief review of group theoretic quantization procedure, we
apply the quantization scheme to the phase space, $M=X \times \R^2$, corresponding to the punctured
plane, $X$. Particularly, we find the
canonical Lie group, $\mathscr{G}$, corresponding to the phase space, $M=X \times \R^2$, to
be $\mathscr{G} = \R^2 \rtimes (SO(2)\times \R^+)$. We establish an algebra homomorphism between
the Lie algebra corresponding to the canonical group, $\mathscr{G} = \R^2 \rtimes (SO(2)\times \R^+)$,
and the smooth functions, $f\in C^{\infty}(M)$, in the phase space, $M=X \times \R^2$. Making
use of this homomorphism and unitary representation of the canonical group,
$\mathscr{G} = \R^2 \rtimes (SO(2)\times \R^+)$, we deduce a quantization map that maps a subspace of
classical observables, $f\in C^{\infty}(M)$, to self-adjoint operators on the Hilbert
space, $\mathscr{H}$, which is the space of all square integrable functions
on $X=\mathbb{R}^2-\{0\}$ with respect to the measure $\dd \mu = \dd \phi\dd\rho/(2\pi\rho)$.
}
\maketitle
\section{Introduction}
Canonical quantization in ${\mathbb R}^n$ is a crude process of associating every classical variable, corresponding to a
classical dynamical system, with self-adjoint operators on a Hilbert space \cite{dirac1981principles, sakurai1967advanced}. The vectors in the Hilbert space are called
the quantum states corresponding to the quantized version of the given classical dynamical system. The canonical quantization
works well in ${\mathbb R}^n$, primarily, due to the global vector space structure of
${\mathbb R}^n$ and
the associated phase space \cite{von2018mathematical,wald1994quantum}. Also, the canonical quantization largely
depends on the global coordinate chart in the phase space that one chooses to work with. Therefore, the canonical quantization may not be suitable for quantizing classical dynamical systems with generic phase spaces, particularly when they possess nontrivial topologies. The simplest examples where the global coordinate chart does not exist are the phase spaces corresponding to a particle that is
constrained to move on a circle, ${\mathbb S}^{1}$, sphere, ${\mathbb S}^{2}$, and cylinder, ${\mathbb S}^{1}\times \R$; and hence quantizing this simple dynamical system turns out to be nontrivial \cite{bojowald2000group,e2021particle,kastrup2011quantization}.

Unlike canonical quantization, group theoretic quantization scheme, proposed by Isham, turns out to be useful
when the phase space of a given dynamical system possess nontrivial topology, and consequently lacks a global
coordinate chart \cite{isham}. Particularly, suppose the given classical dynamical system has a phase space that is a smooth symplectic manifold, $(M,\omega)$. The smooth functions, $f:M\to \R$, in the phase space, $(M,\omega)$, represent the classical observables of the dynamical system, and the time evolution of these classical observables is given by
\begin{equation}
 \dot{f} = \{f,H\},
\end{equation}
where $H$ denotes the Hamiltonian, a smooth function on $M$, corresponding to the classical dynamical system, and $\{\cdot, \cdot\}$ signifies the Poisson bracket on the phase space, $(M,\omega)$. In this setting, group theoretic quantization of the phase space, $(M,\omega)$, refers to a map $\hat{\cdot}:C\subset C^\infty(M) \to \mathscr{O}(\mathscr{H})$, where $\mathscr{O}(\mathscr{H})$ is the algebra of self-adjoint linear operators on the Hilbert space $\mathscr{H}$, and $C$ is a subalgebra of $C^\infty(M)$, such that the poisson bracket algebra is related to the commutator algebra in the following way
\begin{equation}\label{quant}
 \{f,g\} \mapsto -\f{i}{\hbar} [\hat{f}, \hat{g}].
\end{equation}
Particularly, the quantization map, $\widehat{\cdot}$, is primarily based on constructing a canonical group, $\mathscr{G}$ -- a Lie group corresponding to the phase space, $(M,\omega)$, such that the group action is symplectomorphic. The Lie algebra corresponding to the canonical group, $\mathscr{G}$, bridges the Poisson bracket algebra of classical variables, and the algebra of unitary operators on a Hilbert space, $\mathscr{H}$. Specifically, the action of elements of canonical group, $\mathscr{G}$, on the elements of the phase space, $(M,\omega)$, induces a map between the Lie algebra of the canonical group, $\mathscr{G}$, and $C^\infty(M)$; and there are Lie group representations, $\mathcal{U}:G\to U(\mathscr{H})$, that gives a quantization map from a subalgebra of $C^\infty(M)$ and operators in $U(\mathscr{H})$. Note that the unitary representations of the canonical group, $\mathscr{G}$, may not be unique, and therefore the group theoretic quantization scheme assigns each unitary representation of the canonical group, $\mathscr{G}$, a particular quantization map. Another important remark about group theoretic quantization is that the group product on the canonical group, $\mathscr{G}$, determines the "Weyl-like" relations on the unitary operators on a Hilbert space, $\mathscr{H}$. Using these weyl-like relations we determine the commutation relations between the operators on the Hilbert space, $\mathscr{H}$, which is in contrast to the usual canonical quantization scheme employed in quantizing ${\mathbb R}^{n}$ \cite{isham}.

In this paper, we quantize the punctured plane, \ie, ${\mathbb R}^2-\{0\}$, employing Isham's group
theoretic quantization
scheme, where the phase space is $S^1\times \R^3$. By looking at the symmetries of the phase space corresponding to the punctured plane, \ie, rotation, scaling of position coordinates, and translation in canonically conjugate momentum coordinates, and the requirement that group action on the phase space, $S^1\times \R^3$, should be symplectomorphic, we identify the canonical group, ${\mathscr G}$,to be $\R^2\rtimes(\text{SO}(2,\R)\times \R^+)$. We discuss further details in the forthcoming sections. Particularly, in Sec.(\ref{Sec:II}), we briefly review the group theoretic quantization scheme, as discussed in \cite{isham}. In Sec.(\ref{Sec:III}), as an illustrative example, we quantize $\R^2$ employing Isham's group theoretic quantization scheme. In Sec.(\ref{Sec:IV}), we discuss in detail the quantization of punctured plane, $\R^2-\{0\}$, using Isham's group theoretic quantization scheme. Finally, we conclude and summarize our results in Sec.(\ref{Sec:V}).
\section{Isham's group theoretic quantization: A brief review}\label{Sec:II}
\begin{definition}
 A smooth manifold, $M$, is said to be the {\bf homogeneous space} of a Lie group, $G$, if the Lie group, $G$, admits a
 transitive action on $M$, i.e. $\forall x,y\in M$ there is a $g\in G$ such that $y=g\cdot x$.
\end{definition}

Suppose the given phase space, $(M,\omega)$, is a smooth symplectic manifold, then quantizing it employing the group theoretic
quantization procedure involves the following steps below:
\vskip 10pt\noindent
$\bullet$ \textbf{Step-I}: The first step is to construct a Lie group, $\mathscr{G}$, called the canonical group of the phase
space, $(M,\omega)$, such that the smooth manifold, $M$, is a homogeneous space of $\mathscr{G}$, and the action of each
element of the group, $\mathscr{G}$, on the elements of $M$ is symplectomorphic.
\vskip 10pt\noindent
$\bullet$ \textbf{Step-II}: In this step, one constructs a map, $\gamma:\mathfrak{g} \to \text{VF}(M)$, where $\mathfrak{g}$
denotes the Lie algebra corresponding to the canonical group, $\mathscr{G}$, and $\text{VF}(M)$ denotes the set of all vector
fields in the smooth manifold, $M$. For this purpose, consider a map $\R\to \mathscr{G}$, that is given explicitly as
$t \mapsto \exp(tA)$, where $A\in \mathfrak{g}$, and $\exp(tA) \in \mathscr{G}$. Therefore, the image of the map,
$t \mapsto \exp(tA)$, is a one-parameter subgroup of the canonical group, $\mathscr{G}$, and a vector field, $\gamma^A \in \text{VF}(M)$,
is defined as
\begin{equation}
 \gamma^A_x(f) = \dv{f(\exp(-tA)x)}{t}\bigg|_{t=0},\ \text{where}\ x\in M,\ f\in C^\infty( M).
\end{equation}
These vector fields, $\gamma^A$, are called fundamental vector fields of the Lie algebra,
$\mathfrak{g}$, corresponding to the canonical group, $\mathscr{G}$.
Also, since the map, $\gamma:\mathfrak{g} \to \text{VF}(M)$, is a homomorphism (see Appendix. \ref{app:C}),
we obtain
\begin{equation}
 [\gamma^A, \gamma^B] = \gamma^{[A,B]}.
\end{equation}
\vskip 10pt\noindent
$\bullet$ \textbf{Step-III}:
\begin{definition}
 A vector field, $\xi$, is said to be a Hamiltonian vector field, if there exists a function,
 $f\in C^\infty(M)$,
 in the phase space, $(M,\omega)$, such that
 \begin{equation}
  \dd f(X) = \omega(\xi, X),\qq{where} X\in \text{VF(M)}.
 \end{equation}
\end{definition}
\begin{figure}\label{fig:1}
\[\begin{tikzcd}
	{\mathfrak{g}} && {C^\infty(M)} \\
	& {\text{HamVF}(M)}
	\arrow["P", dashed, from=1-1, to=1-3]
	\arrow["\gamma"', from=1-1, to=2-2]
	\arrow["\xi", from=1-3, to=2-2]
\end{tikzcd}\]
\caption{The fundamental vector fields, $\gamma$, are necessarily to be Hamiltonian vector fields, so that an isomorphism from the Lie algebra, $\mathfrak{g}$, into the Poisson algebra, $C^{\infty}(M)$,
can be established.}
\end{figure}
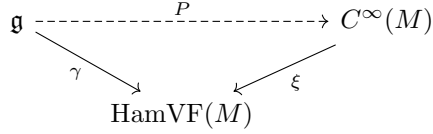
The goal in this step is to construct a homomorphism from the Lie algebra, $\mathfrak{g}$, corresponding to the canonical
group, $\mathscr{G}$, into the classical observables, $f \in C^\infty(M)$. From Step-II above, we know that every element,
$A$, in the Lie algebra, $\mathfrak{g}$, corresponds to a fundamental vector field, $\gamma^A \in \text{VF(M)}$. Moreover,
every function, $f \in C^\infty( M)$, in the phase space, $(M,\omega)$, corresponds to a Hamiltonian vector field, $\xi_{f}$ \cite{loringtu}. Therefore, for the purpose of constructing an isomorphism from the Lie algebra, $\mathfrak{g}$,
into the classical observables, $C^\infty(M)$, we require the fundamental vector fields, $\gamma^{A}$, to be the
Hamiltonian vector fields. Consequently, this restricts the canonical group, $\mathscr{G}$, to be such that all the
fundamental vector fields, $\gamma^A$, corresponding to the elements, $A$, in the Lie algebra, $\mathfrak{g}$ to be
Hamiltonian vector fields. Supposing a canonical group, $\mathscr{G}$, satisfies this condition,
then define a map $P:\mathfrak{g} \to C^\infty(M)$, which is given by $A \mapsto P_A$, where the classical observable,
$P_A$, corresponds to the field $-\gamma^{A}$ (see Figure. \ref{fig:1}).
\vskip 10pt\noindent
$\bullet$ \textbf{Step-IV}:
\begin{definition}
The action of a group, $G$, on a manifold, $M$, is said to be \textit{effective}, if $g_{1},g_{2}\in G$
and $x\in M$, then
 \begin{equation}
  g_{1}x = g_{2}x \implies g_{1} = g_{2}.
 \end{equation}
\end{definition}
Since a Hamiltonian vector field could correspond to more than one element in the Lie algebra, $\mathfrak{g}$, the map,
$\gamma$, will not be one-to-one. Consequently, the map, $P$, will not be one-to-one. This stems from the fact that the
action of two distinct one-parameter sub-groups on the phase space, $(M,\omega)$, could lead to the same elements. This
can be avoided if we restrict the action of the phase space, $(M,\omega)$, to be effective.

More formally, if $A\in \text{Ker}(\gamma)$ then
\begin{equation}
        \dv{f(\exp(-tA)x)}{t}\bigg|_{t=0} = 0, \quad \forall x\in M.
\end{equation}
In other words, in some small enough neighborhood of $t = 0$, one obtains
\begin{equation}\label{temp}
 \exp(-tA)x = x.
\end{equation}
Therefore, imposing the condition that the action of the canonical group, $\mathscr{G}$, on the phase space,
$(M,\omega)$, is effective, forces the map, $\gamma$, to be injective. In other words, supposing that the
action is effective then the only solution to Eq.\eqref{temp} is $A=0$, which means $\text{Ker}(\gamma) = \{0\}$. 
Note that effectiveness can be assumed without loss of generality, since one can always make a group
action effective by taking the quotient of the canonical group, $\mathscr{G}$, with the normal subgroup,
$H=\{g\in \mathscr{G}\ :\ gx =x,\ x\in M\}$. In general, it is enough just to impose the group action to
be effective outside a discrete subgroup of the canonical group, $\mathscr{G}$, which allows one to use
group actions of the covering spaces corresponding to the canonical group, $\mathscr{G}$, on the phase space,
$M$. 
\vskip 10pt\noindent
$\bullet$ \textbf{Step-V:}
If $\xi_{f}$ be the Hamiltonian vector field corresponding to a classical observable, $f \in C^{\infty}(M)$, then
\begin{equation}
 \xi_{P_{[A,B]}} = -\gamma^{[A, B]} = -[\gamma^A, \gamma^B] = -[\xi_A, \xi_B] = \xi_{\{P_A, P_B\}}.
\end{equation}
If $\xi_{f}=\xi_{g}$, then $f-g={\rm constant}$, then one can write
\begin{equation}
 P_{[A,B]} - \{P_A, P_B\} = z(A,B), \quad z(A,B)\in \R.
\end{equation}
Therefore, the map, $P$, is not a homomorphism in general. If $z(A,B)=0$, then the linear map, $P$, is an
algebra isomorphism into the Poisson algebra, $C^\infty(M,\R)$, and is called a momentum map. In the cases
where the factor, $z(A,B)$, cannot be made to vanish by adding a constant to $P$, the Lie algebra is extended
to $\mathfrak{h} = \mathfrak{g} \oplus \R$, with Lie bracket
\begin{equation}
 [(A,r), (B,s)] = ([A,B], z(A,B)).
\end{equation}
Consider the map, $P':\mathfrak{h}\to C^\infty(M,\R)$, that is given by $P'_{(A,r)} = P_A + r$, where $r\in \R$,
then one can show that
\begin{equation*}
 \{P'_{(A,r)}, P'_{(B,s)}\} = \{P_A, P_B\} = P_{[A,B]} + z(A,B) = P'_{([A,B], z(A,B))} = P'_{[(A,r), (B,s)]},
\end{equation*}
which means the map, $P'_{(A,r)}$, is a homomorphism. Moreover, the canonical group, $\mathscr{G}$, is simply replaced
by the unique simply connected Lie group, $H$, which has the Lie algebra $\mathfrak{h}$.
\vskip 10pt\noindent
$\bullet$ \textbf{Step-VI:} Finally, let $\mathscr{H}$ be a separable Hilbert space, on which let
$\text{Unitary}(\mathscr{H})$ be the group of unitary operators, and let the map
$U: \mathscr{G} \to \text{Unitary}(\mathscr{H})$ be a weakly continous, irreducible unitary representation of the
canonical group, $\mathscr{G}$, such that
\begin{equation}
 U(\exp(A)) = e^{-i\hat{K}_A},
\end{equation}
where $\hat{K}_A$ is a self adjoint operator on the Hilbert space, $\mathscr{H}$. Furthermore, since
\begin{align}
 U((\exp(tA)\exp(sB))(\exp(-tA)\exp(-sB))) = \l(e^{-it\hat{K}_A}e^{-is\hat{K}_B}\r)\l(e^{it\hat{K}_A}e^{is\hat{K}_B}\r),
\end{align}
and employing the identity
\begin{equation}\label{bkh}
        (e^{tA}e^{sB})(e^{-tA}e^{-sB}) = e^{ts[A,B] + \text{higher order terms}},
\end{equation}
we obtain
\begin{align}
 U(\exp(ts[A,B] + \text{higher order})) &= e^{-i(ts[\hat{K}_A, \hat{K}_B] + \text{higher order})}, 
\end{align}
so that
\begin{equation}\label{eqn:KAKB}
 [\hat{K}_A, \hat{K}_B] = \hat{K}_{[A,B]}.
\end{equation}

As we know that the momentum map, $P:\mathfrak{g} \to C^{\infty}(M)$, is a Lie algebra homomorphism, and
Eq.(\ref{eqn:KAKB}) implies that a map $\hat{K}:\mathfrak{g} \to \text{SelfAdj}(\mathscr{H})$ is a homomorphism,
where $\text{SelfAdj}(\mathscr{H})$ is the algebra of Hermitian operators on the Hilbert space, $\mathscr{H}$,
the quantization map can be given by associating $P_A$ with $-i\hat{K}_A$. Therefore, it follows to
\begin{equation}
 \{P_A, P_B\} = P_{[A,B]} \mapsto -i\hat{K}_{[A,B]} = -i[\hat{K}_A, \hat{K}_B],
\end{equation}
which means the quantization map satisfies Eq.\eqref{quant}.
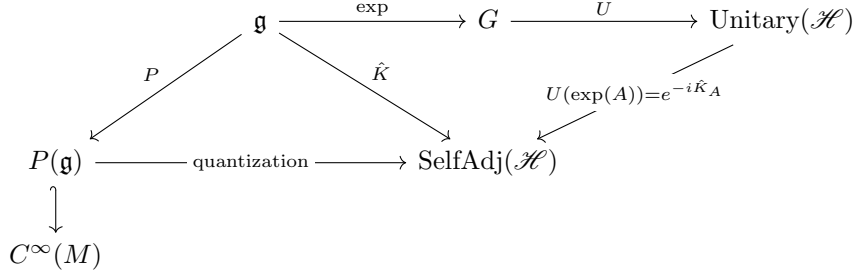
\begin{figure}
\[\begin{tikzcd}
	&& {\mathfrak{g}} && G && {\text{Unitary}(\mathscr{H})} \\
	\\
	{P(\mathfrak{g})} &&&& {\text{SelfAdj}(\mathscr{H})} \\
	{C^\infty(M)}
	\arrow["\exp", from=1-3, to=1-5]
	\arrow["P"', from=1-3, to=3-1]
	\arrow["{\hat{K}}", from=1-3, to=3-5]
	\arrow["U", from=1-5, to=1-7]
	\arrow["{U(\exp(A)) = e^{-i\hat{K}_A}}"{description}, from=1-7, to=3-5]
	\arrow["{\text{quantization}}"{description}, from=3-1, to=3-5]
	\arrow[hook', from=3-1, to=4-1]
\end{tikzcd}\]
\caption{A schematic overview of Isham's group theoretic quantization procedure.}
\end{figure}
Furthermore, the canonical group, $\mathscr{G}$, can have more than one weakly continuous, irreducible unitary
representations. In general, these irreducible representations can be deduced using Mackey's theory of induced representations \cite{mackey1968induced}.

Note that not all classical systems can be quantized employing this method, since the Hamiltonian of the classical
system could not be in the set of classical observables, $P(\mathfrak{g})$, and topologically it is possible to
construct symplectic manifolds which do not admit a transitive action of any Lie group. A necessary condition for
the existence of a transitive Lie group action is that the connected components of the manifold
should be diffeomorphic to each other, and in most physical cases such conditions are satisfied.
\section{Group theoretic Quantization of $\mathbb{R}^2$}\label{Sec:III}
Before quantizing the punctured plane, it is instructive to consider quantizing $\R^2$ using Isham's group theoretic
quantization procedure.
Let $X= \R^2$ be the configuration space of a classical system, then the corresponding phase space is $M = \R^4$, with
the usual symplectic form, $\omega$, given by \cite{arnold2006mathematical, abraham2008foundations}
\begin{align}
 \omega = \dd x \wedge \dd p_x + \dd y\wedge \dd p_y.
\end{align}
Following the quantization scheme stated out in the previous section, we quantize $\R^2$ as shown below:
\vskip 10pt\noindent
$\bullet$ \textbf{Step-I}:
Considering the canonical group, $\mathscr{G}=(\R^4,+)$, to be $\R^4$, and the action of the canonical group,
$\mathscr{G}$, on the phase space, $M=\R^4$, be
\begin{align}
 (\vb{u}, \vb{v})\cdot (\vb{x}, \vb{p}) = (\vb{x} + \vb{u}, \vb{p} - \vb{v}),
\end{align}
where the bold letters denote the corresponding tuples. It can be easily shown that
the group action of the canonical group, $\mathscr{G}=\R^4$, is symplectic and transitive on the phase
space, $M=\R^4$. The Lie algebra, $\mathfrak{g} = \R^4$, corresponding to the canonical group,
$\mathscr{G}=\R^4$, is related to the exponential map, $\exp: \mathfrak{g} \to \mathscr{G}$, as
\begin{align}
 \exp(\vb{a},\vb{b}) = (\vb{a}, \vb{b}).
\end{align}
\vskip 10pt\noindent
$\bullet$ \textbf{Step-II}: Since the fundamental vector fields corresponding to the elements of the
Lie algebra $\mathfrak{g}=\R^4$, is
\begin{equation}
 \gamma^{(\vb{a}, \vb{b})}_{(\vb{x},\vb{p})}(f) = \dv{f(\exp(-tA)(\vb{x},\vb{p}))}{t}\bigg|_{t=0},\
 f\in C^\infty( M),
\end{equation}
we find the fundamental vector fields as
\begin{align}
 \gamma^{(\vb{a}, \vb{b})} = -a_1 \pdv{}{x} - a_2 \pdv{}{y} + b_1\pdv{}{p_x} + b_2 \pdv{}{p_y}.
\end{align}
\vskip 10pt\noindent
$\bullet$ \textbf{Step-III}: For the fundamental vector fields, $\gamma^{(\vb{a}, \vb{b})}$, corresponding
to the elements in the Lie algebra, $\mathfrak{g}=\R^4$, to be Hamiltonian vector fields, there should
exist a function, $f\in C^{\infty}(\R^4)$, such that
\begin{align}
 \dd f(X) = \omega\l(\gamma^{(\vb{a}, \vb{b})},X\r),
\end{align}
where $X$ is an arbitrary vector field in the phase space, $M=\R^4$. For the symplectic form, $\omega$,
corresponding to the phase space, $M=\R^4$, we obtain

\begin{align}
 \dd f(X) = -a_1 \dd p_x(X) -b_1 \dd x(X) -a_2 \dd p_y(X) - b_2\dd y(X),
\end{align}
which leads to the function, $f(\vb{x}, \vb{p}) \in C^{\infty}(\R^4)$, explicitly as
\begin{align}
 f(\vb{x}, \vb{p}) = -\vb{a}\cdot \vb{p} - \vb{b}\cdot\vb{x}.
\end{align}
Consider the map $P$ from $\R^4 \to C^\infty(\R^4)$ given by
\begin{align}
 P_{(\vb{a}, \vb{b})} = \vb{a}\cdot \vb{p} + \vb{b}\cdot\vb{x},
\end{align}
so that $P_{(\vb{a}, \vb{b})}=-\gamma^{(\vb{a}, \vb{b})}$.
\vskip 10pt\noindent
$\bullet$  \textbf{Step-IV}: The group action is effective, since
$(\vb{x}-\vb{u},\vb{p}-\vb{v}) = (\vb{x},\vb{p})$ if and only if $\vb{u}=0$ and
$\vb{v}=0$. Therefore, it straightforwardly follows that $\gamma$ is a one-one homomorphism.
\vskip 10pt\noindent
$\bullet$ \textbf{Step-V}: The map, $P$, is homomorphic only when
\begin{align}
 P_{[(\vb{a},\vb{b}),(\vb{a'},\vb{b'})]} - \{ P_{(\vb{a}, \vb{b})},  P_{(\vb{a}', \vb{b}')} \}=0,
\end{align}
but we find that
\begin{align}
 P_{[(\vb{a},\vb{b}),(\vb{a'},\vb{b'})]} - \{ P_{(\vb{a}, \vb{b})},  P_{(\vb{a}', \vb{b}')} \}
 &= -\vb{b}\cdot \vb{a}' + \vb{b}'\cdot\vb{a},
\end{align}
which concludes that the map, $P$, is not homomorphic. To make it homomorphic, we extend the Lie
algebra, $\mathfrak{g}=\R^4$, to the algebra $\mathfrak{h} = \R^4 \oplus \R$, with the Lie bracket
\begin{align}
 [(\vb{a}, \vb{b}, r), (\vb{a}', \vb{b}', s)] = (0, 0, \vb{b}\cdot \vb{a}' - \vb{b}'\cdot\vb{a}),
\end{align}
which is nothing but the Heisenberg algebra. The unique Lie group, $H$, corresponding to the Lie
algebra, $\mathfrak{h}$, is the Heisenberg group with the group product
\begin{align}
 (\vb{u}, \vb{v}, t)\cdot (\vb{u}', \vb{v}', t') = (\vb{u+ u'}, \vb{v+ v'}, t+t'
 + \f{1}{2} (\vb{u}'\cdot\vb{v} - \vb{u}\cdot\vb{v}')).
\end{align}
Furthermore, the new momentum map, $P'_{(\vb{a}, \vb{b}, r)}$, is given by
\begin{align}
 P'_{(\vb{a}, \vb{b}, r)} = \vb{a}\cdot \vb{p} + \vb{b}\cdot\vb{x} + r.
\end{align}
\vskip 10pt\noindent
$\bullet$ \textbf{Step-VI}: Consider an irreducible unitary representation,
$\mathcal{U}: \mathscr{G} \to \text{Unitary}({\mathscr{H}})$, of the canonical group,
$\mathscr{G}=H_{5}$, \ie,
$(\vb{u}, \vb{v}, t) \mapsto \mathcal{U}(\exp(\vb{a}, \vb{b}, r))$.
Defining the operators $U, V$, and $W$ as
\begin{align}
 U(\vb{a}) \equiv \mathcal{U}(\exp(\vb{a}, 0,0)),\ V(\vb{b}) \equiv \mathcal{U}(\exp(0, \vb{b},0)),
 \And W(r) \equiv \mathcal{U}(\exp(0,0,r)),
\end{align}
we find that they satisfy the Weyl relations
\begin{align}\label{eqn:R2Weyl}
 U(\vb{a}) U(\vb{a}') &= U(\vb{a} + \vb{a}')\\
 V(\vb{b})V(\vb{b}') &= V(\vb{b}+\vb{b}')\\
 U(\vb{a})V(\vb{b}) &= V(\vb{b})U(\vb{a})W(-\vb{a}\cdot\vb{b}).
\end{align}
Furthermore, due to Stone's theorem \cite{prugovecki1982quantum}, the operators $U, V$, and $W$ can be
expressed as the exponential of a densely defined self-adjoint operators as
\begin{align}\label{eqn:UVW}
 U(\vb{a}) = e^{-i(a_1\hat{p_x} + a_2\hat{p_y})},\ V(\vb{b}) = e^{-i(b_1\hat{x} + b_2\hat{y})},
 \And W(r) = e^{-ir\hat{z}}.
\end{align}
Substituting the expressions in Eq.(\ref{eqn:UVW}) in the Weyl relations in Eq.(\ref{eqn:R2Weyl}),
we deduce the commutation relations corresponding to the self-adjoint operators to be
\begin{align}
 [\hat{x}_i, \hat{x}_j] &= 0,\ [\hat{p_i}, \hat{p_j}] = 0,\ [\hat{x}_i, \hat{p_i}] = i\hat{z}\\
 [\hat{x}_i, \hat{z}] &= 0, \And [\hat{p}_i, \hat{z}] = 0.
\end{align}
    
Moreover, from the Stone-Von Neumann theorem \cite{prugovecki1982quantum}, we know that the only irreducible
unitary representation of the Heisenberg group, $H_5$, is of the form
\begin{align}\label{rep}
 \mathcal{U}(\exp(\vb{a}, \vb{b}, r)) \psi(\vb{x}) = e^{-i\vb{b}\cdot\vb{x}-i\mu r} \psi(\vb{x}
 -\mu \vb{a}),
\end{align}
where $\psi \in L^2(\R^2)$, with the Lebesgue measure, $\dd x\dd y$, on $\R^2$. The parameter,
$\mu\in \R$, is a free parameter, and for each $\mu$, we have unitarily inequivalent representations.
Furthermore, the representation of the self adjoint operators, $\hat{x}_i$, and $\hat{p}_i$, can be
found from Eq.\eqref{rep} as
\begin{align}
 \hat{x}_i \psi = x\psi,\ \hat{p}_i \psi = -i\mu \pdv{\psi}{x_i}, \And \hat{z}\psi = \mu \psi.
\end{align}
\section{Group theoretic quantization of the punctured plane}\label{Sec:IV}
Let us now turn to the case of quantizing the punctured plane using group theoretic quantization procedure.
Let the punctured plane, $X = \R^2-\{0\}$, with the usual subspace topology be the configuration space corresponding
to a classical system. Since the punctured plane is a submanifold of $\R^2$, and is diffeomorphic to $S^1 \times \R^+$,
the punctured plane, $X = \R^2-\{0\}$, becomes a multiply-connected space with fundamental group $\Z$.
The phase space, \ie, the cotangent bundle, of the punctured plane, $X$, can be straightforwardly found to be
$M = T^*X = X \times \R^2$. Any point in the phase space, $M$, can be written as $(x,y,p_x, p_y)$, where $(x,y)\in X$,
and $(p_x, p_y)\in \R^2$. The symplectic form, $\omega$, in the phase space, $M$, is induced from $\R^4$, which can be
explicitly written as
\begin{equation}
 \omega = \dd x \wedge \dd p_x + \dd y \wedge \dd p_y.
\end{equation}
Quantizing the phase space, $M=X \times \R^2$, corresponding to the punctured plane, $X = \R^2-\{0\}$,
using group theoretic quantization procedure is carried out as shown below:
%
\vskip 10pt\noindent
$\bullet$ \textbf{Step-I}: Consider the canonical group, $\mathscr{G}$,
as $\mathscr{G} = \R^2 \rtimes (SO(2)\times \R^+)$, with the group product given by
\begin{align}
 (\vb{u},A_\t,\lambda)(\vb{u}', A_{\t'}, \lambda')
 = (\vb{u}+ \lambda^{-1}A_\t\vb{u}', A_{\t+\t'}, \lambda \lambda'),
\end{align}
where $A_{\t}$ is
\begin{equation}
 A_{\t} \equiv
 \begin{pmatrix}
  \cos\t & -\sin\t\\
  \sin\t & \cos\t
 \end{pmatrix}.
\end{equation}
Note that this group product is nothing but the semidirect product induced by the homomorphism,
$\Phi: SO(2) \times \R^+ \to \text{Aut}(\R^2)$, which maps $(A, \lambda) \mapsto \lambda^{-1}A$. Since the semidirect
product of two Lie groups is also a Lie group, one can conclude that the canonical group,
$\mathscr{G}=\R^2 \rtimes (SO(2)\times \R^+)$, is a Lie group as both $\R^2$ and $SO(2)\times \R^+$ are Lie groups.

We consider the action of the canonical group, $\mathscr{G} = \R^2 \rtimes (SO(2)\times \R^+)$,
on the phase space, $M=X \times \R^2$, as
\begin{equation}
 (\vb{u}, A_\t, \lambda)\cdot (\vb{x}, \vb{p}) = (\lambda A_\t \vb{x}, \lambda^{-1} A_\t \vb{p} - \vb{u}),
\end{equation}
and it can be shown straightforwardly that the group action is transitive, and preserves the symplectic form,
$\omega$, corresponding to the phase space, $M=X \times \R^2$. Furthermore, the Lie algebra,
$\mathfrak{g}$, associated with
the canonical group, $\mathscr{G} = \R^2 \rtimes (SO(2)\times \R^+)$,
is a four dimensional algebra, whose elements are represented by
$(b_1,b_2,\t,r) \in \mathfrak{g}$. The basis elements of the Lie algebra, $\mathfrak{g}$, is related to the elements
of the canonical group, $\mathscr{G} = \R^2 \rtimes (SO(2)\times \R^+)$,
through the exponential map, $\exp: \mathfrak{g} \to \mathscr{G}$, as
\begin{align}
 \exp(b_1,0,0,0) &= (b_1,0,I_2,1), \\
 \exp(0,b_2,0,0) &= (0,b_2,I_2,1), \\
 \exp(0,0,\t,0) &= (0,0,A_\t,1), \\
 \exp(0,0,0,r) &= (0,0,I_2,e^r).
\end{align}
Employing the identity in Eq.\eqref{bkh}, the Lie bracket on the Lie algebra, $\mathfrak{g}$, is determined to be
\begin{align}
 [(b_1, b_2, \t, r),(b'_1,b'_2,\t',r')]
 = \l(\t' b_2 - \t b'_2 + r'b_1 - rb'_1, \t b'_1 - \t' b_1 + r'b_2 - rb'_2, 0, 0\r).
\end{align}
\vskip 10pt\noindent
$\bullet$ \textbf{Step-II}: Since the fundamental vector fields corresponding to the elements
of the Lie algebra, $\mathfrak{g}$, is
\begin{equation}
 \gamma^{(\vb{b},\t,r)}_{(\vb{x},\vb{p})}(f) = \dv{f(\exp(-t(\vb{b},\t,r))(\vb{x},\vb{p}))}{t}
 \bigg|_{t=0},\ f\in C^\infty( M),
\end{equation}
we find the fundamental vector fields in the phase space, $M$, as
\begin{align}
 \gamma^{(\vb{b}, \t, r)} &= (\t y -r x) \pdv{}{x} + (-\t x- r y) \pdv{}{y}
 + (r p_x + \t p_y + b_1) \pdv{}{p_x} \nn \\
 &+ (r p_y - \t p_x + b_2)\pdv{}{p_y}.
\end{align}
\vskip 10pt\noindent
$\bullet$ \textbf{Step-III}: In order to show that the fundamental vector fields, $\gamma^{(\vb{b},\t,r)}$, are
Hamiltonian vector fields, one has to construct an observable, say, $f\in C^\infty(M)$, such that
$\dd f(\cdot) = \omega(\gamma^{(\vb{b},\t,r)},\cdot)$. For this purpose, we consider the classical observable
\begin{align}
 f(x,y,p_x,p_y) = -(rxp_x + ryp_y + b_1 x+ b_2 y -\t yp_x + \t xp_y),
\end{align}
which evidently satisfies the required property, and this proves that the vector field, $\gamma^{(\vb{b},\t,r)}$,
is a Hamiltonian vector field. Employing the explicit form of the classical observables, $f$, we define the map,
$P:\mathfrak{g}\to C^\infty(M)$, as
\begin{align}
 P_{(\vb{b}, \t, r)} = r\vb{x}\cdot \vb{p} + \vb{b}\cdot \vb{x} + \t\vb{x}\cdot C \vb{p},
\end{align}
where
\begin{align*}
 C \equiv
 \begin{pmatrix}
  0 & 1\\
  -1 & 0
 \end{pmatrix},
\end{align*} 
so that 
\begin{align}
 \xi_{P_{(\vb{b},\t,r)}} = -\gamma^{(\vb{b},\t,r)}. 
\end{align}
\vskip 10pt\noindent
$\bullet$ \textbf{Step-IV}: The action of the canonical group,
$\mathscr{G} = \R^2 \rtimes (SO(2)\times \R^+)$, on the phase space, $M=X \times \R^2$, is effective,
since $(\lambda A_\t \vb{x}, \lambda^{-1}A_\t\vb{p}-\vb{u}) = (\vb{x}, \vb{p})$, if and only if
$(A_\t - \lambda^{-1}I_2)\vb{x} = 0$ and $(A_\t-\lambda I_2)\vb{p} = \vb{u}$. Since the matrix, $A_\t$, has no
positive real eigenvalues, unless $\t=0$, in which case the eigenvalue is $1$, so that 
$A_\t = I_2$, $\lambda=1$,
and $\vb{u}=0$. Furthermore, as discussed in Step-IV of section \ref{Sec:II}, the effectiveness of the action
of the canonical group, $\mathscr{G} = \R^2 \rtimes (SO(2)\times \R^+)$, on the phase space,
$M=X \times \R^2$, implies that the map, $\gamma: \mathfrak{g} \to \text{HamVF}(M)$, is one-one.
\vskip 10pt\noindent
$\bullet$ \textbf{Step-V}: Similarly, as discussed in Step-V of section \ref{Sec:II}, the map
$P:\mathfrak{g}\to C^\infty(M)$ is a homomorphism, if
$\{P_{(\vb{b}, \t,r)},P_{(\vb{b}', \t',r')} \} = P_{[(\vb{b}, \t, r),(\vb{b}', \t', r)]}$,
which can be verified as below:
\begin{align}
 \{P_{(\vb{b}, \t,r)},P_{(\vb{b}', \t',r')} \}
 &= \pdv{P_{(\vb{b}, \t,r)}}{\vb{x}} \pdv{P_{(\vb{b}', \t',r')}}{\vb{p}}
 - \pdv{P_{(\vb{b}, \t,r)}}{\vb{p}}\pdv{P_{(\vb{b}', \t',r')}}{\vb{x}}, \nonumber\\
 &= \l[(r\vb{p} + \vb{b} + \t C\vb{p})\cdot(r'\vb{x} - \t' C\vb{x})\r]
 - \l[(r'\vb{p} + \vb{b}' + \t' C\vb{p})\cdot(r\vb{x} - \t C\vb{x})\r], \nonumber\\
 &= \vb{x}\cdot \l(r'\vb{b} - r\vb{b} - \t C\vb{b} + \t'C\vb{b}'\r), \nonumber\\
 &= P_{(\t' b_2 - \t b'_2 + r'b_1 -rb'_1, \t b'_1 - \t b_1 + r'b_2 - rb'_2)}, \nonumber\\
 &= P_{[(\vb{b}, \t, r),(\vb{b}', \t', r)]}.
 \end{align}
Therefore, it is evident that the map, $P:\mathfrak{g}\to C^\infty(M)$, is a momentum map.
\vskip 10pt\noindent
$\bullet$ \textbf{Step-VI}: Let $\mathcal{U}:\mathscr{G}\to \text{Unitary}(\mathscr{H})$ be an irreducible
unitary representation of the canonical group, $\mathscr{G} = \R^2 \rtimes (SO(2)\times \R^+)$, which maps
$(\vb{u}, A_\t, \lambda) \mapsto \mathcal{U}(\vb{u}, A_\t, \lambda)$. Defining the operators,
\begin{align}
 U(\t, r) =\mathcal{U}(\exp(0, \t, r)) \And V(\vb{b}) = \mathcal{U}(\exp(\vb{b}, 0, 0)),
\end{align}
where $r = \log(\lambda)$, we obtain the Weyl-like relations as
\begin{align}\label{eq:weyl-like}
 \begin{split}
  U(\t,r)U(\t',r')&= U(\t+ \t', r+r'),\\
  V(\vb{b})V(\vb{b}') &=V(\vb{b} + \vb{b}'),\\
  U(\t, r)V(\vb{b}) &=V(e^{-r}A_\t\vb{b})U(\t,r).
 \end{split}
\end{align}
Since the operators, $U$ and $V$, can be further written as a product of unitary operators, we essentially
have four one parameter group of unitary operators. Using Stone's theorem \cite{prugovecki1982quantum} on each of
them, we write
\begin{align}
 U(\t, r) = e^{-i(\t\hat{\pi}_1 + r\hat{\pi}_2)/\hbar} \And V(\vb{b}) = e^{-i(b_1 \hat{c} + b_2 \hat{s})/\hbar}.
\end{align}
Since $SO(2)\times \R^+$ and $\R^2$ are both Abelian groups, it follows that
\begin{align}
 [\hat{\pi}_1, \hat{\pi}_2] = 0, ~~ \text{and} ~~ [\hat{c}, \hat{s}] = 0.
\end{align}
Furthermore, employing the identity in Eq. \eqref{bkh}, we find the commutation relations as
\begin{align}
 [\hat{s}, \hat{\pi}_1] &= i\hat{c}, ~~~~ [\hat{c}, \hat{\pi}_1] = -i\hat{s}\label{com1}, \\
 [\hat{s}, \hat{\pi}_2] &= i\hat{s}, ~~~~ [\hat{c}, \hat{\pi}_2] = i\hat{c}\label{com2}.
\end{align}
Therefore, the quantization map for the phase space, $M=X \times \R^2$,
with the canonical group, $\mathscr{G} = \R^2 \rtimes (SO(2)\times \R^+)$, is given by
\begin{align}
 P^{(\vb{b},0,0)} = b_1 x + b_2 y \mapsto -i(b_1 \hat{c} + b_2\hat{s}),
\end{align}
and
\begin{align}
 P^{(0,\t,r)} = rxp_x + ryp_y - \t yp_x + \t x p_y \mapsto -i(\t\hat{\pi}_1 + r\hat{\pi}_2).
\end{align}
     
Let $\mathscr{H}$ be the space of all square integrable functions on $X=\mathbb{R}^2-\{0\}$
with respect to the measure
$\dd \mu = \dd \phi\dd\rho/(2\pi\rho)$. A weakly continuous, irreducible unitary representation of the
canonical group, $\mathscr{G}$, in $\text{Unitary}(L^2(X))$ can be written as
\begin{subequations}\label{eq:reps}
\begin{align}
 U(\t, \lambda) \psi(\p, \rho) &= \psi((\p-\t)\text{mod}(2\pi), \lambda^{-1}\rho), \label{eq:repa}\\
 V(\vb{b}) \psi(\p, \rho) &= e^{-i(b_1\cos\p + b_2\sin\p)\rho/\hbar} \psi(\p, \rho),\label{eq:repb}
\end{align}
\end{subequations}
where $\psi\in L^2(X)$. Therefore, the representation of the self-adjoint operators can be deduced by
expanding both sides of the Eqs. \eqref{eq:repa} and \eqref{eq:repb} as a power series as
\begin{align}
 \hat{c} \psi(\p, \rho) &= (\rho \cos\p) \psi(\p, \rho),
 ~~~~ \hat{s} \psi(\p, \rho) = (\rho \sin\p) \psi(\p, \rho), \\
 \hat{\pi}_1 \psi(\p, \rho) &= -i\hbar\pdv{\psi}{\phi},
 ~~~~ \hat{\pi}_2 \psi(\p, \rho) = -i \hbar\rho \pdv{\psi}{\rho}.
\end{align}
    
As discussed in Step-VI of section \ref{Sec:II}, if the canonical group,
$\mathscr{G} = \R^2 \rtimes (SO(2)\times \R^+)$, admits universal
covering groups, then the concerned covering group could also be considered as a new canonical group associated
with the phase space, $M$ \cite{isham}. Consider the universal covering group,
$\tilde{\mathscr{G}} = \R^2 \rtimes (\R\times \R^+)$, with a universal covering map,
$\pi:\tilde{\mathscr{G}} \to \mathscr{G}$, that has the group product
\begin{align}
 (\vb{u}, \t, \lambda)\cdot (\vb{u}', \t', \lambda')
 = (\vb{u} + \lambda^{-1} A_{\t}\vb{u}', \t+\t', \lambda\lambda'),
\end{align}
and $\pi(\vb{u},\t,\lambda)= (\vb{u}, A_\t, \lambda)$. The action of the universal covering group,
$\tilde{\mathscr{G}}$, on the phase space, $M$, is induced from the action of the base group,
$\mathscr{G} = \R^2 \rtimes (SO(2)\times \R^+)$, on
the phase space, $M$, \ie $\ \tilde{g}\cdot (\vb{x},\vb{p}) = \pi(g)\cdot x$, where $\tilde{g}\in \tilde{\mathscr{G}}$.
The covering group, $\tilde{\mathscr{G}}$, acts effectively everywhere, except for the discrete subgroup $2\pi \Z$ of
$\tilde{\mathscr{G}}$. Employing the twisted representations \cite{isham}, we find a family of inequivalent
unitary representations, parametrized by $\alpha \in \R$, of the new
canonical group, $\tilde{\mathscr{G}}$, which are given by
\begin{subequations}\label{neq_rep}
\begin{align}
 U(\t, \lambda) \psi(\p, \rho) = e^{-i\a \t} \psi((\p -\t)\text{mod}(2\pi), \lambda^{-1}\rho), \\
 V(\vb{b}) \psi(\p, \rho) = e^{-i(b_1\cos(\p) + b_2\sin(\p))\rho/\hbar} \psi(\p, \rho),
\end{align}
\end{subequations}
and which results in the operators as
\begin{align}
 \hat{\pi}_1 \psi(\p, \rho) = -i\hbar\pdv{\psi}{\phi} + \hbar\a \psi,
 &\And \hat{\pi}_2 \psi(\p, \rho) = -i \hbar\rho \pdv{\psi}{\rho}, \\
 \hat{c}\psi(\phi,\rho) = \rho \cos(\phi)\psi(\phi,\rho),
 &\And \hat{s}\psi(\phi,\rho) = \rho \sin(\phi)\psi(\phi,\rho).
\end{align}
\section{Conclusion}\label{Sec:V}
We quantize the punctured plane, \ie, ${\mathbb R}^2-\{0\}$, using Isham's group theoretic
quantization procedure. We find the phase space, $M$, corresponding to the punctured plane to be
$M =\R^2\times (S^1\times \R)$, with a natural symplectic form, $\omega$, and identify the canonical group, ${\mathscr G}$, to be ${\mathscr G}=\R^2\rtimes(\text{SO}(2,\R)\times \R^+)$, and determine the corresponding Lie algebra, $\mathfrak{g}$, to be a four dimensional vector space with the Lie bracket 
\begin{align*}
 [(b_1, b_2, \t, r),(b'_1,b'_2,\t',r')]
 = \l(\t' b_2 - \t b'_2 + r'b_1 - rb'_1, \t b'_1 - \t' b_1 + r'b_2 - rb'_2, 0, 0\r).
\end{align*}
We find the fundamental vector fields, $\gamma^A$, corresponding to the lie algebra element $A\in \mathfrak{g}$, and show that
all the fundamental vector fields are Hamiltonian vector fields with respect to the symplectic form, $\omega$.
We also note that the action of the canonical group, ${\mathscr G}=\R^2\rtimes(\text{SO}(2,\R)\times \R^+)$,
on the phase space, $M =\R^2\times (S^1\times \R)$, is effective, and thus the algebra homomorphism
$\gamma: \mathfrak{g} \to \text{HamVF}(M)$ is one-one. 
This allows one to define the momentum map, $P$, which we determine to be an algebra homomorphism from
$\mathfrak{g}$ into $C^\infty(M)$.
Using an explicit representation of the canonical group,
${\mathscr G}=\R^2\rtimes(\text{SO}(2,\R)\times \R^+)$, in the group of unitary operators on the
Hilbert space, $\mathscr{H} = L^2(S^1\times \R)$, we determine "Weyl-like" relations. Using Stone's
theorem and Baker-Campbell-Hausdorff formula, we find self-adjoint operators,
$\hat{c}, \hat{s}, \hat{\pi}_1,$ and $ \hat{\pi}_2$ on the Hilbert space, $\mathscr{H}$,
and the commutation relations between these are found, as shown in Eqs. \eqref{com1} and \eqref{com2}.
Finally, using the fact that the universal covering group, $\tilde{\mathscr{G}}$, of the canonical group,
${\mathscr G}=\R^2\rtimes(\text{SO}(2,\R)\times \R^+)$, is also admissible as a canonical group, and using
the theory of twisted representations \cite{isham}, we find the representations of the universal
covering group, $\tilde{\mathscr{G}}$.
\bibliography{sn-bibliography}
\begin{appendices}
\section{$\gamma:\mathfrak{g}\to \text{VF}(M)$ is homomorphic}
\label{app:C}
\begin{definition}
    Let $F:M\to N$ be a smooth map between manifolds $M$ and $N$, and let $X,\tilde{X}$ be vector fields
    in manifolds, $M$ and $N$, respectively, then the vector fields, $X$ and $\tilde{X}$, are $F-$related, if
    \begin{align}
      F_{*,p}(X_p) = \tilde{X}_{F(p)},
    \end{align}
    for all $p\in M$, where $F_{*,p}$ is the induced push-forward map between the tangent space, $T_p(M)$, at $p\in M$, and the tangent space, $T_{F(p)}(N)$, at $F(p)\in N$.
\end{definition}
\vskip 10pt\noindent
\begin{definition}
    Let $\mathscr{G}$ be a Lie group with Lie algebra, $\mathfrak{g}$, and $A\in \mathfrak{g}$. For any group element, $g\in \mathscr{G}$, let $r_g:\mathscr{G}\to \mathscr{G}$ be the map defined as $g'\mapsto g'g$ and let $(r_g)_*$ be the induced map on the tangent spaces $T_e(\mathscr{G})\to T_g(\mathscr{G})$. Then we define a right-invariant vector field $R^A$ in the Lie group, $\mathscr{G}$, as
    $$R^A_{g} = (r_g)_*(A)$$
\end{definition}
\vskip 10pt\noindent
We state two well-known results \cite{loringtu} below, that are useful in our proof.
\begin{theorem}\label{Tu1}
Let $F:M\to N$ be a smooth map between the manifolds, $M$ and $N$, and let $X$ and $\tilde{X}$ be vector fields in the manifolds, $M$ and $N$,  respectively, then the vector fields, $X$ and $\tilde{X}$, are $F-$related, if and only if
\begin{align}
    X(f\circ F) = \tilde{X}(f) \circ F,
\end{align}
for all $f\in C^\infty(N)$.
\end{theorem}
\vskip 10pt\noindent
\begin{theorem}\label{Tu2}
Let $F:M\to N$ be a smooth map between the manifolds, $M$ and $N$, and let $X$ and $\tilde{X}$ be vector fields in the manifolds, $M$ and $N$, respectively. If the vector fields, $X$ and $\tilde{X}$,
are $F-$related, and if the vector fields, $Y$ and $\tilde{Y}$, are $F-$related, then the vector fields,
$[X,Y]$ and $[\tilde{X},\tilde{Y}]$, are also $F-$related.
\end{theorem}

\vskip 10pt
In this appendix we show that the map, $\gamma:\mathfrak{g}\to \text{VF}(M)$, from the Lie algebra, $\mathfrak{g}$, of the Lie group, $\mathscr{G}$, to the Poisson algebra of vector fields on the phase space, $M$, is a Lie algebra homomorphism. This essentially involves showing that the map, $\gamma$, is
a linear map, and it preserves the Lie bracket expression.
The key idea underlying the proof is to show that for any Lie algebra element, $A\in \mathfrak{g}$, the Hamiltonian vector fields, $\gamma^A$, and the right invariant vector field, $R^{-A}$ are $t^x-$related, where we define $t^x:\mathscr{G}\to M$ to be the map between the canonical group, $\mathscr{G}$, and the phase space, $M$. We sketch out the proof in Theorem \ref{thm:tx-related} below.
\vskip 10pt\noindent
\begin{theorem}\label{thm:tx-related}
    Let $\mathscr{G}$ be a Lie group with Lie algebra, $\mathfrak{g}$, $A\in \mathfrak{g}$, and $t^x:\mathscr{G}\to M$ be a map between the canonical group, $\mathscr{G}$, and the phase space, $M$,
then the vector fields, $\gamma^A$ and $R^{-A}$, are $t^x-$related for any $x\in M$, and
the vector fields, $[\gamma^A,\gamma^B]$ and $[R^{-A},R^{-B}]$, are also $t^x-$related.
\end{theorem}
\begin{proof}
    Let $g\in \mathscr{G}$, and let $f:M\to \R$ be a smooth map, then
    \begin{align}
        R^{-A}_g(f\circ t^x) &= (r_g)_*(-A)(f\circ t^x), \\
        &= -A(f\circ t^x \circ r_g), \\
        &= \dv{f\circ t^x\circ r_g(\exp(-tA))}{t}\bigg|_{t=0}, \\
        &= \dv{f\circ t^x(\exp(-tA)g)}{t}\bigg|_{t=0}, \\
        &= \dv{f(\exp(-tA)gx)}{t}\bigg|_{t=0}, \\
        &= \gamma^A_{gx}(f) = \gamma^A(f)\circ t^x(g),
    \end{align}
where in the third equality we made use of the fact that the tangent to the curve,
$\exp(-tA)$, at $t=0$ is $-A$. Therefore, due to Theorem \ref{Tu1}, it is evident that
the vector fields, $\gamma^A$ and $R^{-A}$, are $t^x-$related for any $x\in M$, and it follows
from Theorem \ref{Tu2} that the vector fields, $[\gamma^A,\gamma^B]$ and
$[R^{-A},R^{-B}]$ are also $t^x-$related.
\end{proof}

To show linearity of the map, $\gamma:\mathfrak{g}\to \text{VF}(M)$, it is enough to note that
\begin{align}
    R_g^{A+B} = r_{g*}(A+B) = r_{g*}(A)+r_{g*}(B) = R_g^{A}+ R_g^B.
\end{align}
Due to Theorems \ref{Tu1} and \ref{thm:tx-related}, we obtain
\begin{align}
    \gamma^{A+B}(f)\circ t^x(g) &= \gamma^{A}(f)\circ t^x(g) + \gamma^B(f)\circ t^x(g), \\
    \gamma^{A+B}(f)(gx) &= \gamma^{A}(f)(gx) + \gamma^B(f)(gx),
\end{align}
for all group elements $g\in \mathscr{G}$, phase space points $x\in M$, and smooth functions
$f\in C^\infty(M)$; and setting $g=e$, we obtain
\begin{align}
    \gamma^{A+B}(f)(x) &= \gamma^{A}(f)(x) + \gamma^B(f)(x).
\end{align}
This implies that
\begin{align}
 \gamma^{A+B}(f)=\gamma^{A}(f) + \gamma^B(f),
\end{align}
for all smooth functions $f\in C^\infty(M)$, which implies that $\gamma:\mathfrak{g}\to \text{VF}(M)$ is linear.

Similarly, to show that the map, $\gamma$, preserves the Lie brackets, it is enough to show that $[R^{-A},R^{-B}] = R^{[B,A]} = R^{-[A,B]}$, which is true since $R^A = -\iota_*L^A$, and Proposition 16.7 in Ref. \cite{loringtu}, where $\iota:\mathscr{G}\to \mathscr{G}$ is the map $g\mapsto g^{-1}$, and $L^A$ is the left-invariant vector field, $(\ell_g)_*(A)$, on the Lie group, $\mathscr{G}$ \cite{loringtu}. 
Similarly, following the same line of arguments as in the previous paragraph, one can complete
the proof of the claim that $\gamma:\mathfrak{g}\to \text{VF}(M)$ is a Lie algebra homomorphism
by employing Theorem \ref{thm:tx-related}.
\end{appendices}

\end{document}